\documentclass[runningheads]{llncs}
\usepackage[T1]{fontenc}
\usepackage{graphicx}
\usepackage{booktabs}
\usepackage{amsmath,amssymb}
\usepackage{booktabs}
\usepackage[ruled,vlined,linesnumbered]{algorithm2e}
\usepackage{hyperref}
\usepackage{bbding} % envelope mark for the corresponding author
\newcommand{\ttc}{\mathrm{TTC}}
\newcommand{\A}{\mathcal{A}}
\newcommand{\N}{\mathcal{N}}
\newcommand{\R}{\mathcal{R}}
\newcommand{\PO}{\mathcal{P}}
\usepackage{color}
\usepackage{tikz}
\usetikzlibrary{positioning}

\begin{document}
\title{Inverting Top Trading Cycles: Characterizing the Endowments that
Implement a Pareto-Optimal Allocation}
\titlerunning{Inverting Top Trading Cycles}

% If the paper title is too long for the running head, you can set
% an abbreviated paper title here
%
\author{Bhavik Dodda\inst{1}\orcidID{0009-0003-2434-4625}\and
Garima Shakya\inst{2}\orcidID{0009-0005-0612-5677}}
\authorrunning{B. Dodda and G. Shakya}
% First names are abbreviated in the running head.
% If there are more than two authors, 'et al.' is used.
% Bhavik Dodda is the corresponding author.
%
\institute{Sardar Vallabhbhai National Institute of Technology, Surat, India\\
\email{bhavikdodda22@gmail.com} \and
Indian Institute of Technology Palakkad, Palakkad, India\\
\email{garima@iitpkd.ac.in}\thanks{Corresponding author.}}
\maketitle              % typeset the header of the contribution
\begin{abstract}
In one-sided matching with strict ordinal preferences, the Top
Trading Cycles (TTC) mechanism maps every initial endowment to a
Pareto-optimal (PO) allocation, and many distinct endowments converge
to the same outcome. We study the \emph{inverse} problem: given a PO
allocation $\sigma$, characterize and enumerate its preimage
$f^{-1}(\sigma)$, the set of endowments from which TTC produces
$\sigma$. We characterize $f^{-1}(\sigma)$ via ordered partitions of
the agents and present a sound and complete enumeration algorithm,
invTTC. The preimage carries direct economic meaning:
$|f^{-1}(\sigma)|/n!$ equals the probability that random serial
dictatorship selects $\sigma$, and $f^{-1}(\sigma)$ is the basin of
attraction of $\sigma$, quantifying robustness to the assignment of
initial property rights. We analyze the enumeration's complexity and
report experiments on preimage structure.
\keywords{One-sided matching \and Top Trading Cycles \and Pareto
optimality \and Random serial dictatorship \and Enumeration.}
\end{abstract}

% =====================================================================
\section{Introduction}
% =====================================================================
 
In one-sided matching problems, such as assigning students to hostel
rooms, each of $n$ agents holds a strict ordinal preference over $n$ objects, and
monetary transfers are not permitted. The Top Trading Cycles (TTC)
mechanism of Shapley and Scarf~\cite{shapley1974} is the canonical
solution: starting from an initial endowment (an assignment of one
object to each agent), it repeatedly executes trading cycles and
terminates at an allocation that is \emph{Pareto optimal} (PO) and, for strict
preferences, the unique core allocation of the induced housing
market~\cite{roth1977}.
 
TTC is a \emph{function of the endowment}: different endowments may
lead to different PO outcomes, and many
distinct endowments converge to the \emph{same} outcome. Writing
$f(\sigma') = \ttc(\sigma')$, the map $f$ partitions the $n!$
endowments into equivalence classes, one per PO allocation. We ask:
\emph{given a PO allocation $\sigma$, what exactly is the preimage
$f^{-1}(\sigma)$
of endowments from which TTC yields $\sigma$, and how is it
enumerated?} The preimage has economic content in three ways.
 
\smallskip\noindent\textbf{RSD probabilities.} By the equivalence of
Abdulkadiro\u{g}lu and S\"onmez~\cite{abdulka1998}, TTC from a
uniformly random endowment induces the same distribution over
allocations as random serial dictatorship (RSD); hence
$|f^{-1}(\sigma)|/n!$ \emph{is} the probability that RSD selects
$\sigma$. Computing the RSD assignment-probability matrix is
\#P-complete~\cite{saban2015}, so efficient exact computation of the
outcome probabilities we study is not expected either. On moderate
instances we enumerate $f^{-1}(\sigma)$ exactly from the structure
of $\sigma$.
 
\smallskip\noindent\textbf{Endowment design.} $f^{-1}(\sigma)$ is the
\emph{basin of attraction} of $\sigma$ under TTC: it tells a designer
exactly which initial property rights \emph{implement} a target PO
allocation through decentralized trading, and its size measures the
outcome's robustness to that assignment. Serial dictatorship lists
priority orders. The bijection of Abdulkadiro\u{g}lu and
S\"onmez~\cite{abdulka1998}
(Corollary~\ref{cor:rsd}) matches orders with endowments.
 
\smallskip\noindent\textbf{Structure of the TTC correspondence.} The
classes $\{f^{-1}(\sigma)\}_{\sigma}$ form the quotient of the set of
assignments under TTC-equivalence. Prior work addresses the
\emph{image} side of the map: computing PO
matchings~\cite{abraham2004} and counting reachable
houses~\cite{asinowski2016}; Pareto optimality has further been studied under uncertain preferences~\cite{aziz2019}, quotas and popularity~\cite{cseh2022},
and fairness objectives~\cite{biro2021,hosseini2024}.
Abdulkadiro\u{g}lu and S\"onmez~\cite{abdulka1998} prove that, for each Pareto-efficient allocation, the endowments whose TTC outcome is that allocation are in bijection with the serial-dictatorship orders that produce it. That result does not give conditions on an endowment $\sigma'$ which decide whether $\sigma' \in f^{-1}(\sigma)$.
 
\smallskip\noindent\textbf{Contributions.} We give a declarative
characterization (Theorem~\ref{thm:char}): $\sigma' \in
f^{-1}(\sigma)$ if and only if the agents admit an ordered partition
satisfying a sequential top-choice condition and an
endowment-closure condition. We present \textsf{invTTC}, a sound and
complete enumeration algorithm for $f^{-1}(\sigma)$. We relate
$|f^{-1}(\sigma)|$ to RSD probabilities (Corollary~\ref{cor:rsd}),
analyze the enumeration's complexity, identifying polynomial-delay
enumeration as a concrete open problem, and report an empirical
study of preimage sizes on random instances. We do \emph{not} give a method for
enumerating the full Pareto frontier: under strict preferences the PO
allocations coincide with the serial-dictatorship outcomes over all
priority orders~\cite{abdulka1998,abraham2004}, which enumerates the
frontier more simply.\footnote{The full version of the paper is available here:
\url{https://arxiv.org/pdf/2604.21306}.}
The object of this paper is the preimage
$f^{-1}(\sigma)$.
\section{Model and Notation}
An instance is a triple $H = (\N, \R, O)$, where
$\N = \{1,\dots,n\}$ is a set of agents, $\R = \{r_1,\dots,r_n\}$ a set
of rooms, and $O = (o_1,\dots,o_n)$ a profile of strict ordinal
preferences of agents over rooms. An \emph{allocation} is a bijection
$\sigma : \N \to \R$; $\A$ denotes the set of all $n!$ allocations. An
allocation is PO if no other allocation makes
some agent strictly better off and no agent worse off; $\PO \subseteq
\A$ denotes the set of PO allocations.
 
Given an initial endowment $\sigma' \in \A$, TTC proceeds in rounds: in
each round, every remaining agent points to her most preferred
remaining room and every remaining room points to its current owner;
all directed cycles are executed (each agent in a cycle receives the
room she points to), and the involved agents and rooms are removed.
For strict preferences the outcome is a uniquely determined PO
allocation. We refer to an allocation used as the input of TTC as an
\emph{endowment}, and call $\sigma'^{-1}(r)$ the
\emph{$\sigma'$-owner} of room $r$. We write $f : \A \to \PO$,
$f(\sigma') = \ttc(\sigma')$,
and for $\sigma \in \PO$,
% \[ 
$
f^{-1}(\sigma) \;=\; \{\sigma' \in \A \mid f(\sigma') = \sigma\}.
$
% \]
Since $f$ is total, the sets $\{f^{-1}(\sigma)\}_{\sigma \in \PO}$
partition $\A$. Moreover, every PO allocation is a \emph{fixed point} of $f$: an agent's endowed room can leave the market only in her own round, so TTC assigns each agent a room she weakly prefers to her endowment; for $\sigma \in \PO$ no agent can strictly improve, and preferences being strict, $f(\sigma) = \sigma$. Hence $f$ is onto $\PO$, and $f^{-1}(\sigma) \ni \sigma$ is the basin of attraction of $\sigma$. 
 
\section{A Characterization of the TTC Preimage}

Our main structural result characterizes membership in
$f^{-1}(\sigma)$ without reference to the execution of TTC.

\begin{theorem}[Characterization of $f^{-1}(\sigma)$]\label{thm:char}
Let $\sigma \in \PO$ and $\sigma' \in \A$. Then
$f(\sigma') = \sigma$ if and only if there exists an ordered partition
$(N_1, N_2, \dots, N_k)$ of $\N$ such that for every
$j \in \{1,\dots,k\}$:
% \begin{itemize}
\emph{(C1)} \textbf{(sequential top choice)} for every agent
$i \in N_j$, the room $\sigma(i)$ is $i$'s most preferred room in
$\R \setminus \sigma(N_1 \cup \dots \cup N_{j-1})$; \emph{(C2)} \textbf{(endowment closure)}
$\sigma'(N_j) = \sigma(N_j)$.
% \end{itemize}
\end{theorem}

\begin{proof}
($\Leftarrow$, sufficiency.) With strict preferences, every maximal
sequence of trading-cycle executions terminates in the unique core
allocation of the market~\cite{shapley1974,roth1977}, so the TTC
outcome does not depend on the order in which available cycles are
executed. We may therefore choose the order, and we induct on the
number of layers $k$ (the case $k = 0$ is vacuous). By (C1) with
$j = 1$, every agent $i \in N_1$ initially points to $\sigma(i)$; by
(C2), the induced pointer map $i \mapsto \sigma'^{-1}(\sigma(i))$ is a
bijection of $N_1$ onto itself, and the functional graph of a
bijection is a disjoint union of cycles. Hence $N_1$ partitions into
trading cycles of the TTC pointer graph: agent $i_1$ points to
$\sigma(i_1)$, owned by $i_2 = \sigma'^{-1}(\sigma(i_1))$, who points
to $\sigma(i_2)$, and so on until the cycle closes; an agent with
$\sigma'(i) = \sigma(i)$ forms a cycle of length one. Execute these
cycles first. Every $i \in N_1$ receives
$\sigma(i)$, and exactly the rooms $\sigma(N_1)$ leave the market.
What remains is a housing market on agents $\N \setminus N_1$ with
rooms $\R \setminus \sigma(N_1)$ and the restriction of $\sigma'$ as
endowment (well defined by (C2)), for which the partition
$(N_2, \dots, N_k)$ satisfies (C1) and (C2) verbatim: for $j \ge 2$,
removing $\sigma(N_1)$ is exactly what (C1) already prescribes. By the
induction hypothesis every remaining agent receives her
$\sigma$-room as well; hence $f(\sigma') = \sigma$.

($\Rightarrow$, necessity.) Suppose $f(\sigma') = \sigma$, and let
$N_j$ be the set of agents removed in round $j$ of the execution of
TTC on $\sigma'$, for $j = 1,\dots,k$. (C1): when agent $i \in N_j$ is
removed, she receives the room she points to, which is her most
preferred room among those remaining; the rooms remaining at round $j$
are exactly $\R \setminus \sigma(N_1 \cup \dots \cup N_{j-1})$, since
in every earlier round the removed agents received their $\sigma$-rooms.
As $i$ receives $\sigma(i)$, condition (C1) follows. (C2): the rooms
assigned in round $j$ are exactly the rooms owned, under the current
endowment, by the agents on the executed cycles; since cycles trade
rooms among their own members, $\sigma(N_j) = \sigma'(N_j)$. \qed
\end{proof}
 
Theorem~\ref{thm:char} says that an endowment reaches $\sigma$ exactly
when the agents can be layered so that each layer trades only within
itself (C2) and, once earlier layers' rooms are gone, each agent's
target room is her top remaining choice (C1). 
Note that condition (C1) depends only on $(\sigma, O)$ and not on $\sigma'$: the admissible ordered partitions are a property of the target allocation, and (C2) then describes the endowments compatible with each such partition.

\begin{corollary}[RSD probability]\label{cor:rsd}
For every $\sigma \in \PO$, the probability that random serial
dictatorship selects $\sigma$ equals $|f^{-1}(\sigma)|\,/\,n!$.
Equivalently, $|f^{-1}(\sigma)|$ equals the number of priority orders
for which serial dictatorship outputs $\sigma$.
\end{corollary}

\begin{proof}
Immediate from the equivalence theorem of~\cite{abdulka1998}: TTC from
a uniformly random endowment and RSD induce the same distribution over
allocations; matching the point masses at $\sigma$ gives both
statements. \qed
\end{proof}

\section{Enumerating the Preimage: the invTTC Algorithm}
\label{sec:algo}
Theorem~\ref{thm:char} is existential; we now turn it into an
enumeration procedure, \textsf{invTTC}, that generates
$f^{-1}(\sigma)$ by constructing the layers of the ordered partition
in \emph{forward} order $N_1, N_2, \dots$: a room is committed
(``squared'') once its position in the endowment under construction is
fixed, and committing the rooms of earlier layers exposes the
candidate rooms of the next layer.
\begin{algorithm}[ht!]
\caption{\textsf{invTTC}$(\sigma, O)$: enumerate $f^{-1}(\sigma)$}
\KwIn{PO allocation $\sigma$, preference profile $O$. \textbf{Output:} $Z = f^{-1}(\sigma)$.}
% \KwOut{$Z = f^{-1}(\sigma)$}
$\tau_0 \leftarrow$ rooms at their $\sigma$-positions, all unmarked;
$\tau_0 \leftarrow \textsf{dressup}(\tau_0, \emptyset)$\;
initialize queue $Q \leftarrow \{\tau_0\}$, visited $V \leftarrow
\emptyset$, output $Z \leftarrow \emptyset$\;
\While{$Q \neq \emptyset$}{
  pop tagged state $s$ from $Q$; \textbf{if} $s \in V$ \textbf{then
  continue}; add $s$ to $V$\;
  \If{all rooms in $s$ are squared}{
    append the induced endowment to $Z$; \textbf{continue}\;
  }
  \ForEach{non-identity permutation $\pi$ of the circled rooms of $s$}{
    apply $\pi$; square every displaced room; re-run
    \textsf{dressup}; push resulting state onto $Q$\;
  }
  \ForEach{nonempty subset $B$ of circled rooms of $s$}{
    square the rooms in $B$ (undisplaced); re-run \textsf{dressup};
    push resulting state onto $Q$\;
  }
}
\Return{$Z$}
\end{algorithm}
The algorithm maintains a \emph{tagged state}: an arrangement of the
$n$ rooms over the $n$ agent positions, in which each room additionally
carries a tag in $\{-1, 0, 1\}$, meaning respectively \emph{unmarked}
(its layer has not yet been determined), \emph{circled} (eligible to
join the next layer), and \emph{squared} (its position in the endowment
is fixed). A useful invariant is that \emph{every unsquared room sits
at the position of its $\sigma$-owner}: a room leaves its
$\sigma$-position only by being displaced, and displaced rooms are
squared immediately.

\smallskip\noindent\textbf{\textsf{dressup}} (activation): given the
set $D$ of squared rooms, a room $r = \sigma(i)$ is circled iff it is
agent $i$'s most preferred room in $\R \setminus D$, which is
condition (C1) verbatim once $D$ collects the rooms of the layers
built so far.
%\smallskip\noindent
\textbf{\textsf{devour}} (branching): (i) apply
any non-identity permutation $\pi$ of the circled rooms, replacing
the room at each circled position by its $\pi$-image; rooms that
change position are \emph{displaced} and squared
immediately, which is admissible by (C2); or (ii) square any nonempty subset
of circled rooms in place. Then re-run \textsf{dressup}, which may
circle newly exposed rooms, and recurse; an all-squared state is one
endowment of $f^{-1}(\sigma)$. 
Appendix~A restates the procedure as explicit steps and traces it on a worked example.
% \vspace{-0.5ex}
\begin{theorem}[Soundness and completeness of \textsf{invTTC}]
\label{thm:sc}
For every PO allocation $\sigma$,
$\textsf{invTTC}(\sigma, O) = f^{-1}(\sigma)$.
\end{theorem}
\begin{proof}[Proof sketch]
\emph{Soundness.} Fix a terminal state, let $B_1,\dots,B_m$ be the
sets of rooms squared at each squaring step of the branch that reached
it, in chronological order, and let $N_t$ be the agents at whose
positions the rooms of $B_t$ were squared. We claim
$(N_1,\dots,N_m)$ witnesses Theorem~\ref{thm:char} for the induced
endowment. \emph{(C2)}: rooms squared without displacement (subset
branch) sit at their $\sigma$-positions by the invariant above, so
their sub-layer is trivially closed; rooms displaced by a permutation
move only into positions vacated by other displaced rooms, and every
vacated position is the $\sigma$-position of a displaced room (again
by the invariant), hence $B_t = \sigma(N_t)$ in both cases.
\emph{(C1)}: a room is squared only while circled, and at step $t$ the
squared set is $D = B_1 \cup \dots \cup B_{t-1} =
\sigma(N_1 \cup \dots \cup N_{t-1})$ by (C2), so the circling
condition of \textsf{dressup} is (C1) verbatim. By
Theorem~\ref{thm:char}, the induced endowment lies in
$f^{-1}(\sigma)$.

\emph{Completeness.} Let $\sigma' \in f^{-1}(\sigma)$ and let
$(N_1,\dots,N_k)$ with $S_j = \sigma(N_j)$ be the round partition of
TTC on $\sigma'$, which satisfies (C1) and (C2) by the
($\Rightarrow$) direction of Theorem~\ref{thm:char}. Process the
layers in \emph{forward} order. Inductively, once all rooms of
$S_1,\dots,S_{j-1}$ are squared, $D = \sigma(N_1\cup\dots\cup
N_{j-1})$, so by (C1) every room of $S_j$ is circled (other rooms may
be circled as well). Choose the permutation that acts on $S_j$ as
$\sigma'$ (well defined since $\sigma'$ permutes $S_j$ over $N_j$ by
(C2)) and fixes every other circled room; the displaced rooms of
$S_j$ are squared at their $\sigma'$-positions, and one further
subset step squares the remaining rooms of $S_j$, namely those of
agents with $\sigma'(i) = \sigma(i)$. Splitting the round $N_j$ into
a displaced and an undisplaced sub-layer in this way is admissible,
since Theorem~\ref{thm:char} does not require layers to be maximal
and both sub-layers inherit (C1) and (C2). Induction over
$j = 1,\dots,k$ reaches a terminal state whose arrangement is exactly
$\sigma'$.

\emph{Termination.} Each branching step strictly increases the number of squared rooms, and the visited set ensures every state is processed once.
\qed
\end{proof}
\noindent\textbf{Complexity.}
Each successor state is constructed in $O(n^2)$ time, dominated by
the preference filtering in \textsf{dressup}, and the number of
outputs, $|f^{-1}(\sigma)|$, spans the full range from $1$ (identical
preferences: TTC is the identity on endowments, so every preimage is
a singleton) to $n!$ (pairwise distinct top choices: the unique PO
allocation's preimage is all of $\A$); the enumeration is thus
inherently output-sensitive. Two caveats. First, the number of
intermediate states explored between consecutive outputs is not
polynomially bounded in general, so \textsf{invTTC} is not known to
achieve polynomial delay; establishing DelayP, IncP, or OutputP
guarantees for preimage enumeration, or ruling them out, is our
main open problem. Second, by Corollary~\ref{cor:rsd} the counting
version is tied to RSD probabilities, whose assignment-matrix form is \#P-complete~\cite{saban2015}, so efficient exact counting of
$|f^{-1}(\sigma)|$ is not expected; good delay is the natural algorithmic target.

\smallskip\noindent\textbf{Empirical Study of Preimage Structure.} 
% \vspace{-2ex}
% \vspace{-2ex}
\begin{table}[h!]
\centering
\small
\setlength{\tabcolsep}{3pt}
\begin{tabular}{@{}rrrrrrrrrr@{}}
\toprule
 & \multicolumn{2}{c}{preimage size} &
 \multicolumn{3}{c}{RSD prob.\ $|f^{-1}(\sigma)|/n!$} &
 \multicolumn{2}{c}{depth} & \multicolumn{2}{c}{tree size} \\
\cmidrule(lr){2-3}\cmidrule(lr){4-6}\cmidrule(lr){7-8}\cmidrule(lr){9-10}
$n$ & max & mean & min & mean & max & mean & max & mean & max \\
\midrule
3 & $3.32$     & $2.92$    & $0.42$   & $0.4866$ & $0.5533$ & $1.03$ & $1.50$ & $8.74$    & $9.72$ \\
5 & $32.73$    & $14.54$   & $0.0463$ & $0.1211$ & $0.2727$ & $2.64$ & $3.83$ & $48.16$   & $98.13$ \\
7 & $671.53$   & $116.28$  & $0.0013$ & $0.0230$ & $0.1332$ & $4.48$ & $5.98$ & $397.29$  & $1994.17$ \\
9 & $24559.06$ & $1394.80$ & $1.15{\times}10^{-5}$ & $0.0038$ & $0.0676$ & $6.41$ & $8.00$ & $4813.64$ & $71909.38$ \\
\bottomrule
\end{tabular}
\caption{Preimage and inverse-search statistics: per-profile
min/mean/max over a profile's preimage classes, averaged over the
same $100$ uniformly random profiles per $n$}
\label{tab:emp}
\end{table}
\vspace{-5ex}

For each $n$ we generated $100$ preference profiles
uniformly at random over strict rankings and computed the full
partition $\{f^{-1}(\sigma)\}_{\sigma \in \PO}$ of $\A$; for every class we also recorded the
\emph{depth} (the deepest BFS level at which a new endowment first
appears) and the \emph{tree size} (the number of distinct tagged
states visited) of the inverse search. Table~\ref{tab:emp}
summarizes the results. Preimage sizes are highly skewed: a few PO
allocations absorb most endowments (at $n = 9$ the largest basin's
RSD probability averages $0.068$, over $5{,}000$ times the
smallest), and the concentration sharpens with $n$: at $n = 3$ the
largest and the mean basin nearly coincide ($3.32$ versus $2.92$
endowments), whereas at $n = 9$ they differ by more than an order of
magnitude. The mean depth grows roughly linearly in $n$, while the per-profile
maximum depth concentrates at $n-1$ (attaining it in every profile at
$n = 9$); tree size grows much faster and is itself skewed: at $n = 9$ the largest tree
averages about $15$ times the mean ($71{,}909$ versus $4{,}814$
states). The reported tree sizes stay a factor of
about three of the corresponding preimage sizes at every
$n$ ($4{,}813.64$ versus $1{,}394.80$ in the mean at $n = 9$), so the \emph{average}
exploration cost per enumerated endowment appears bounded, while
worst-case polynomial delay remains open.
Appendix~\ref{app:time} compares the running time of computing the full set of PO allocations by \textsf{invTTC} with that of TTC on all $n!$ endowments.

\section{Discussion and Conclusion}
Beyond its structural interest, the preimage supports concrete tasks:
computing exact RSD outcome probabilities (Corollary~\ref{cor:rsd}),
identifying every default assignment that TTC trading carries to a
target allocation, and selecting, among PO allocations meeting a
fairness constraint~\cite{biro2021,hosseini2024}, the one most likely
under RSD. Its limitations are the unavoidable worst-case exponential
cost, experiments limited to $n \le 9$, and strict preferences (with
ties, the question must be re-posed under a tie-breaking rule). In
summary, Theorem~\ref{thm:char} characterizes and \textsf{invTTC}
enumerates the endowments that TTC maps to a given PO allocation; the
central open problem is enumeration with polynomial delay.

\newpage
\begin{credits}
\subsubsection{\discintname}
The authors have no competing interests to declare that are relevant to the content of this article.
\subsubsection{AI Declaration.}
The authors did not use generative AI or AI-assisted technologies in preparing this paper.
\end{credits}

\bibliographystyle{splncs04}
\bibliography{references}

\newpage
% =====================================================================
\appendix
% \section{Appendix}
\section{The \textsf{invTTC} Procedure Step by Step}\label{app:walk}
 
This appendix restates the enumeration procedure of
 Section~\ref{sec:algo} as an explicit sequence of steps and depicts the complete
search graph for Example~\ref{ex:run}.

\smallskip\noindent\textbf{Step 1 (Initialization).}
Place every room at the position of its $\sigma$-owner and tag all
rooms unmarked (tag $-1$). Run \textsf{dressup} with $D = \emptyset$:
every room that is its owner's most preferred room overall is circled
(tag $0$). These are the rooms eligible to form the first layer.
 
\smallskip\noindent\textbf{Step 2 (Branch by shuffling).}
For every non-identity permutation $\pi$ of the currently circled
rooms, apply $\pi$ and square every displaced room (tag $1$): its
position
in the endowment under construction is now committed. Rerun
\textsf{dressup} with the enlarged squared set $D$.
 
\smallskip\noindent\textbf{Step 3 (Branch by committing in place).}
For every nonempty subset $B$ of circled rooms, square the rooms of
$B$ without displacement, committing their owners to their
$\sigma$-rooms as endowment. Rerun \textsf{dressup}: removing squared
rooms from the preference lists may expose new top choices, which are
circled and become eligible for the next layer.
 
\smallskip\noindent\textbf{Step 4 (Collect and terminate).}
States are explored in breadth-first order with a visited set. A
state in which every room is squared spells out one endowment of
$f^{-1}(\sigma)$. Every branch strictly increases the number of
squared rooms, so the search terminates.

 \begin{example}\label{ex:run}
Let $n = 3$ with preferences $1\colon r_2 \succ r_1 \succ r_3$;
$2\colon r_1 \succ r_2 \succ r_3$; $3\colon r_1 \succ r_3 \succ r_2$,
and target $\sigma = (1 \mapsto r_2,\, 2 \mapsto r_1,\,
3 \mapsto r_3)$, which is PO; we write endowments as tuples
$(\sigma'(1), \sigma'(2), \sigma'(3))$. Initially rooms sit at their
$\sigma$-positions, and \textsf{dressup} circles $r_2$ and $r_1$ (the
top choices of agents $1$ and $2$) but not $r_3$, since agent $3$
prefers $r_1$. \textsf{devour} branches: (i) swapping the two circled
rooms squares both, after which $r_3$ is circled and squared,
yielding $(r_1, r_2, r_3)$; (ii) squaring both circled rooms
undisplaced yields $\sigma$ itself, the fixed point; (iii) squaring
$r_1$ \emph{alone} exposes $r_3$, which \textsf{dressup} newly
circles (agent $3$'s top room among $\R \setminus \{r_1\}$ is
$r_3$), and swapping the circled $r_2$ and $r_3$ yields
$(r_3, r_1, r_2)$. All remaining branches reproduce these states and
are pruned. Hence $f^{-1}(\sigma) = \{(r_1,r_2,r_3),\,
(r_2,r_1,r_3),\, (r_3,r_1,r_2)\}$, which agrees with the ground truth
obtained by running TTC directly on each of the $3! = 6$ endowments.
Note that $(r_3, r_1, r_2)$ is reachable only through
the subset branch, illustrating why it is required for completeness.
Consistent with Corollary~\ref{cor:rsd}, exactly three priority
orders, $(1,2,3)$, $(2,1,3)$, and $(2,3,1)$, make serial dictatorship
output $\sigma$, so the RSD probability of $\sigma$ is $3/6 = 1/2$.
The search graph for this instance appears in
Figure~\ref{fig:tree}.
\end{example}

\begin{figure}[h!]
\centering
\resizebox{\textwidth}{!}{%
\begin{tikzpicture}[
  st/.style={draw, rounded corners, inner sep=3pt, font=\small},
  term/.style={st, line width=1pt},
  el/.style={font=\footnotesize, fill=white, inner sep=1pt},
  ->, >=stealth
]
% ---- states ----
\node[st]   (S0) at ( 6.0, 0)    {$[\,r_2^{\circ}\;\, r_1^{\circ}\;\, r_3\,]$};
\node[st]   (S1) at ( 1.6,-1.8)  {$[\,r_1^{\blacksquare}\;\, r_2^{\blacksquare}\;\, r_3^{\circ}\,]$};
\node[st]   (S2) at ( 4.6,-1.8)  {$[\,r_2^{\blacksquare}\;\, r_1^{\circ}\;\, r_3\,]$};
\node[st]   (S3) at ( 8.8,-1.8)  {$[\,r_2^{\circ}\;\, r_1^{\blacksquare}\;\, r_3^{\circ}\,]$};
\node[term] (T3) at (12.2,-1.8)  {$[\,r_3^{\blacksquare}\;\, r_1^{\blacksquare}\;\, r_2^{\blacksquare}\,]$};
\node[term] (T1) at ( 1.6,-3.7)  {$[\,r_1^{\blacksquare}\;\, r_2^{\blacksquare}\;\, r_3^{\blacksquare}\,]$};
\node[st]   (S4) at ( 5.2,-3.7)  {$[\,r_2^{\blacksquare}\;\, r_1^{\blacksquare}\;\, r_3^{\circ}\,]$};
\node[st]   (S5) at ( 9.6,-3.7)  {$[\,r_2^{\circ}\;\, r_1^{\blacksquare}\;\, r_3^{\blacksquare}\,]$};
\node[term] (T2) at ( 7.4,-5.5)  {$[\,r_2^{\blacksquare}\;\, r_1^{\blacksquare}\;\, r_3^{\blacksquare}\,]$};
\node[font=\small, below=0.5mm of T2] {$= \sigma$};
% ---- edges ----
\draw (S0) -- node[el, pos=0.55] {swap $r_2{\leftrightarrow}r_1$} (S1);
\draw (S0) -- node[el, pos=0.5]  {sq.\,$\{r_2\}$} (S2);
\draw (S0) -- node[el, pos=0.5]  {sq.\,$\{r_1\}$} (S3);
\draw (S0) to[bend left=18] node[el, pos=0.7] {sq.\,$\{r_1,r_2\}$} (S4);
\draw (S1) -- node[el] {sq.\,$\{r_3\}$} (T1);
\draw (S2) -- node[el, pos=0.45] {sq.\,$\{r_1\}$} (S4);
\draw (S3) -- node[el] {swap $r_2{\leftrightarrow}r_3$} (T3);
\draw (S3) -- node[el, pos=0.45] {sq.\,$\{r_2\}$} (S4);
\draw (S3) -- node[el, pos=0.5]  {sq.\,$\{r_3\}$} (S5);
\draw (S3) -- node[el, pos=0.35] {sq.\,$\{r_2,r_3\}$} (T2);
\draw (S4) -- node[el, pos=0.5]  {sq.\,$\{r_3\}$} (T2);
\draw (S5) -- node[el, pos=0.5]  {sq.\,$\{r_2\}$} (T2);
\end{tikzpicture}}
\caption{Search graph of \textsf{invTTC} for Example~\ref{ex:run}.
Each state lists the rooms at positions $1, 2, 3$; superscript
$\circ$ marks circled rooms, $\blacksquare$ squared rooms, and no
mark unmarked rooms. Thick-bordered states are terminal; their
arrangements are exactly $f^{-1}(\sigma)$. Edge labels: \emph{swap} =
non-identity permutation branch, \emph{sq.} = subset-squaring branch.
The edges into the terminal state equal to $\sigma$ show
duplicates pruned by the visited set.}
\label{fig:tree}
\end{figure}

\smallskip
Figure~\ref{fig:tree} shows the resulting search graph for
Example~\ref{ex:run}: nine states, of which three are terminal, and
their arrangements are exactly
$f^{-1}(\sigma) = \{(r_1,r_2,r_3),\allowbreak (r_2,r_1,r_3),\allowbreak
(r_3,r_1,r_2)\}$. The three edges converging on the terminal state
equal to $\sigma$ illustrate duplicate discovery: distinct branches
reach the same state and are pruned by the visited set,
underlying the open delay question of
Section~\ref{sec:algo}.

\section{Running Time against Brute-Force TTC}\label{app:time}

For each $n \in \{3,4,5,6,7,8,9\}$ we drew $100$ strict preference
profiles uniformly at random, $700$ profiles in total. For each
profile we timed two procedures that return the set of all PO
allocations. \textsf{invTTC} builds that set from the inverse search.
The baseline runs TTC once from every endowment, so it makes $n!$
calls. Both procedures use the same TTC routine, and the time is
end-to-end wall-clock time in Python. On all $700$ profiles the two
procedures returned sets of the same size.

Table~\ref{tab:time} gives the mean time in seconds. Brute-force TTC
is faster for $n \in \{3,4,5\}$. \textsf{invTTC} is faster for
$n \in \{6,7,8,9\}$. At $n = 9$ the mean times are $16.41$ seconds
for \textsf{invTTC} and $26.43$ seconds for brute-force TTC.
Figure~\ref{fig:time} plots every run. The gap marked on each block
is the difference of the two means, brute-force time minus
\textsf{invTTC} time.

\begin{table}[t]
\centering
\small
\setlength{\tabcolsep}{5pt}
\caption{Mean wall-clock time, in seconds, over $100$ uniform strict
profiles for each $n$. The last column is brute-force time minus
\textsf{invTTC} time.}
\label{tab:time}
\begin{tabular}{@{}rrrr@{}}
\toprule
$n$ & \textsf{invTTC} & brute-force TTC & difference \\
\midrule
3 & $0.004108$ & $0.000462$ & $-0.003646$ \\
4 & $0.004607$ & $0.001030$ & $-0.003577$ \\
5 & $0.007141$ & $0.004414$ & $-0.002727$ \\
6 & $0.023216$ & $0.030263$ & $0.007047$ \\
7 & $0.146653$ & $0.257793$ & $0.111140$ \\
8 & $1.347437$ & $2.443039$ & $1.095602$ \\
9 & $16.408077$ & $26.429038$ & $10.020961$ \\
\bottomrule
\end{tabular}
\end{table}

\begin{figure}[t]
\centering
\includegraphics[width=\textwidth]{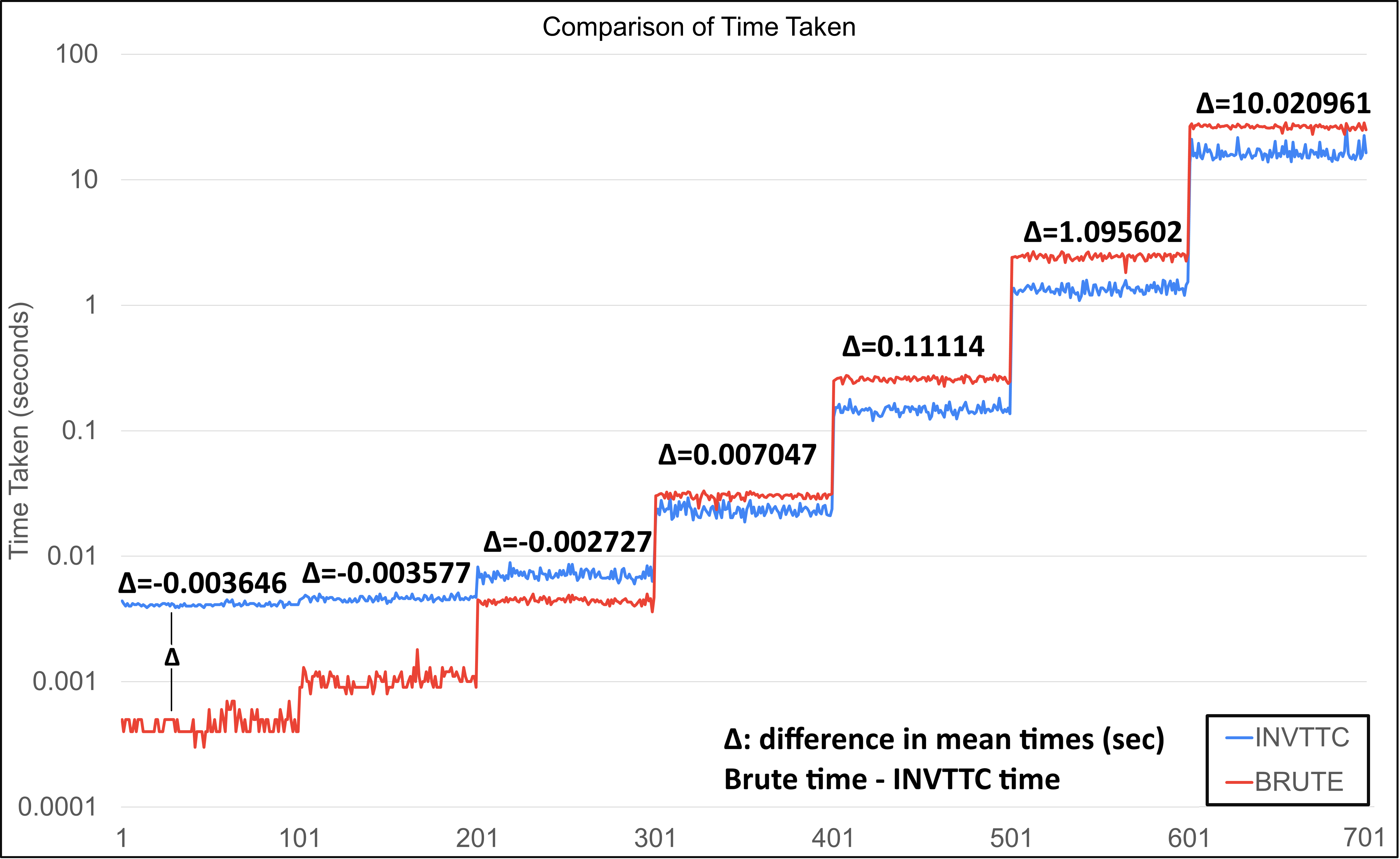}
\caption{Wall-clock time for each of the $700$ profiles, on a
logarithmic scale. Profiles are grouped by $n$, from $n=3$ on the
left to $n=9$ on the right, $100$ profiles in each block.}
\label{fig:time}
\end{figure}

\end{document}